\documentclass[11pt]{article}
\usepackage{amsmath,amsfonts,amssymb,amsthm}
\usepackage{framed}
\usepackage[dvipsnames]{color}
\usepackage[T1]{fontenc}
\usepackage{times} 
\usepackage[cp1250]{inputenc}  
\makeatother
\newtheorem{thm}{Theorem}[section]
\newtheorem{pro}[thm]{Proposition}
\newtheorem{lemma}[thm]{Lemma}
\newtheorem{cor}[thm]{Corollary}
\theoremstyle{definition}
\newtheorem{dfn}[thm]{Definition}
\newtheorem{rem}{Remark}
\numberwithin{equation}{section}
\newcommand{\CA}{\mathcal{A}} 

\newcommand{\C}{\mathbb{C}} 
\newcommand{\IC}{\mathbb{C}} 

\newcommand{\ga}{\gamma} 
\renewcommand{\H}{\mathcal{H}} 
\DeclareMathOperator{\id}{id} 
\newcommand{\acts}{\triangleright} 
\newcommand{\N}{\mathbb{N}} 
\newcommand{\T}{\mathbb{T}} 
\newcommand{\oh}{{\tfrac{1}{2}}} 
\newcommand{\R}{\mathbb{R}} 
\newcommand{\sq}{\unskip\nobreak\kern5pt\nobreak\vrule height4pt width4pt depth0pt} 
\newcommand{\Z}{\mathbb{Z}} 
\renewcommand{\ker}{\hbox{ker\ }} 

\def\<#1,#2>{\langle#1,#2\rangle} 

\def\CB{\mathcal{B}}

\def\CH{\mathcal{H}}

\def\om{\omega}
\def\ti{\tilde}
\def\suq2{{\cal A}(SU_q(2))}
\definecolor{lgray}{rgb}{0.2,0.2,0.2}
\definecolor{shadecolor}{named}{GreenYellow}
\newbox\ncintdbox \newbox\ncinttbox
\setbox0=\hbox{$-$} \setbox2=\hbox{$\displaystyle\int$}
\setbox\ncintdbox=\hbox{\rlap{\hbox
    to \wd2{\kern-.1em\box2\relax\hfil}}\box0\kern.1em}
\setbox0=\hbox{$\vcenter{\hrule width 4pt}$}
\setbox2=\hbox{$\textstyle\int$}
\setbox\ncinttbox=\hbox{\rlap{\hbox
    to \wd2{\kern-.14em\box2\relax\hfil}}\box0\kern.1em}
\newcommand{\ncint}{\mathop{\mathchoice{\copy\ncintdbox}
    {\copy\ncinttbox}{\copy\ncinttbox}
    {\copy\ncinttbox}}\nolimits}
\title{Noncommutative circle bundles \\ and new Dirac operators}
\author{Ludwik D\k{a}browski${}^{1,*}$ and Andrzej Sitarz${}^{2,3,**}$ \\ 
\vbox{
\small
\begin{center}
${}^1$ SISSA (Scuola Internazionale Superiore di Studi Avanzati), \\
via Bonomea 265, 34136 Trieste, Italy \\ \ \\
${}^*$ Partially supported by GSQS 230836 (IRSES, EU) and PRIN 2010-2012 (MIUR, Italy) \\ \ \\
${}^2$ Institute of Physics, Jagiellonian University,\\
Reymonta 4, 30-059 Krak\'ow, Poland \\ \ \\
${}^3$ Institute of Mathematics of the
Polish Academy of Sciences,\\
ul. Sniadeckich 8, Warszawa, 00-950 Poland \\ \ \\
${}^{**}$ Partially supported by MNII grants 189/6.PRUE/2007/7 and N 201 1770 33 \\
\end{center}
}
}
\date{}
\begin{document}
\maketitle
\begin{abstract}
We study spectral triples over noncommutative principal $U(1)$ bundles.
Basing on the classical situation and the abstract algebraic approach, we propose an
operatorial definition for a connection and compatibility between the connection
and the Dirac operator on the total space and on the base space of the bundle.
We analyze in details the example of the noncommutative three-torus viewed
as a $U(1)$ bundle over the noncommutative two-torus and find all connections
compatible with an admissible Dirac operator. Conversely, we find a family of
new Dirac operators on the noncommutative tori, which arise from the
base-space Dirac operator and a suitable connection.
\end{abstract}
\thispagestyle{empty}

\section{Introduction}

The principal $U(1)$ bundles are the simplest and most fundamental examples 
of fibre bundles, often encountered in mathematics and physics.
They are usually equipped both with a connection and a metric,
which are in principle independent of each other, though, an 
interesting situation arises when they are compatible in some natural way.
This is reflected, for example, in the spectral geometry of $U(1)$ bundles,
which in terms of Laplace operator has been studied in \cite{GLP}, 
whereas the analysis of Dirac operator was presented in \cite{Amman,AB}.
In this note we shall extend part of the latter analysis to the analogue of principal $U(1)$ bundles 
in noncommutative geometry, encoding their geometric aspects in terms of spectral triples
\cite{Co1,Co2}.

\section{Spin Geometry of U(1)-bundles}

We suppose that $M$ is a $n+1$ dimensional ($n+1$ odd) compact manifold 
which is the total space of $U(1)$-principal bundle over the $n$-dimensional 
($n$ even) base space $N=M/U(1)$. Moreover assume $M$ is equipped with a 
Riemannian metric $\ti g$ and the $U(1)$ action (free and transitive on fibres) 
is isometric. The base space $N$ carries a unique metric $g$ such that 
the projection $ \pi : (M,\ti g) {\longrightarrow} (N,g)$ is a Riemannian submersion.

We can and shall use a suitable local orthonormal frame (basis) of the tangent space 
$TM$, $e=(e_0, e_1, \ldots, e_n)$, such that $e$ is $U(1)$ invariant and $e_0$ is the  
(normalized) Killing vector field $K$ associated to the $U(1)$-action.  
For simplicity we assume that the fibres are of constant length $2\pi \ell$.  

There exists a unique principal connection 1-form $\om: TM \to \R\approx u(1)$, 
such that $\ker {\om}$ is orthogonal to the fibres for all $m \in M$ with respect to
$\ti g$. Obviously it is given by $\omega = e^0/\ell$, where $(e^0, e^1, \ldots, e^n)$ 
is the dual frame to $e$. Conversely, if we are given a principal connection on the principal
$U(1)$ bundle and a metric on the base space $N$ then there exists a unique 
$U(1)$-invariant metric on $M$, such that the horizontal vectors are orthogonal 
to the fundamental (Killing) vector field $K$ of length $\ell$.

Assume now that $M$ is spin and let $\Sigma M$ be its spinor bundle
(which is hermitian, rank $2^{\frac{n}{2}}$ complex vector bundle). The $U(1)$ 
action either lifts to the spin structure and then to an action
$$\kappa: U(1)\times \Sigma M\to \Sigma M,$$
or to a projective action (up to a sign), i.e. to the action of a non-trivial 
double cover of $U(1)$, which happens to be still $U(1)$ as a group. 

Assuming the former, we have a {\it projectable} spin structure on $M$. 
As explained in \cite{AB} this induces a spin structure on $N$. Conversely,
any spin structure on $N$ canonically induces a projectable spin structure 
on $M$ via a pull-back construction.

We recall that the Dirac operator $\ti D$ on $M$ can be constructed as follows.
Let $\gamma_j$, $j=0, 1,\dots , n$, be the antihermitian matrices in 
$M({2^{\frac{n}{2}}}, \C)$, which satisfy the relations
\begin{equation}
\label{clifrel}
\gamma_j\gamma_k + \gamma_k\gamma_j = - 2 \delta_{jk}.
\end{equation}

Then  Dirac operator $\ti D$ acting on sections of $\Sigma M$ can be
explicitly written as
$$ \ti D = \sum\limits^n_{i=0}  \gamma_i \partial_{e_i}  + 
\frac{1}{4} \sum\limits^n_{i,j,k=0} \ti\Gamma^k_{ij} \gamma_i \gamma_j \gamma_k ,$$ 
where $\ti \Gamma^k_{ij}$ are Christoffel symbols (in the orthonormal 
basis $e$) of the Levi-Civita connection on $M$.
In particular
\begin{equation}
\label{Gamma0}
\begin{aligned}
&-\tilde\Gamma_{ij}^0 =\tilde\Gamma_{i0}^j=\tilde\Gamma_{0i}^j={\ell\over 2}\, d\om(e_i,e_j) , \\
& \tilde\Gamma_{i0}^0=\tilde\Gamma_{0i}^0=\tilde\Gamma_{00}^i=\tilde\Gamma_{00}^0=0.
\end{aligned}
\end{equation}

Since the metric on $M$ is completely characterized by the connection 1-form $\om$, 
the length $\ell$ of the fibres and the metric $g$ on $N$, 
the Dirac operator $\ti D$ on $M$ can be expressed in terms of $\om$, and $g$. 
Conversely, the metric on $N$, the connection $\omega$ and the length of the fibres
can be recovered from $\tilde D$.

Following this line, Ammann and B\"ar showed \cite{AB} how to present the Dirac 
operator $\ti{D}$ as a sum of two first order differential operators on 
$L^2(\Sigma M)$ and a zero order term (endomorphism of the spinor bundle). 

The first operator, called the  {\it vertical} Dirac operator is 
$$ D_v:=\frac{1}{\ell}\ga_0\,\partial_K, $$ 
where
  $$\partial_K (\Psi)(m)= {d \over dt}|_{t=0}\,
    \kappa(e^{-it},\Psi(m\cdot e^{it}))$$
is the Lie derivative of a spinor $\Psi$ along the U(1) Killing field.
Note that $D_v:=\ga_0\,\partial_{e_0} $, where $\partial_{e_0}$ could be interpreted as the Dirac operator 
associated to the typical fibre $S^1\simeq U(1)$,
whereas $\ga_0$ is the Clifford representation of the (normalized) one-form 
$e^0=\ell \omega $.

It follows from (\ref{Gamma0}) that the spinor covariant derivative differs 
from the Lie derivative in the direction of $e_0$:
\begin{equation}\label{nabla0}
\nabla_{e_0} =  \partial_{e_0} +  {\ell\over 4}\,\sum_{j<k} d\om(e_j,e_k)\ga_j\ga_k .
\end{equation}

The description of the second differential operator $D_h$, called a 
{\it horizontal} Dirac operator, uses an orthogonal decomposition of
the Hilbert space into irreducible representations of $U(1)$:
$$ L^2(\Sigma M)=\bigoplus_{k \in \Z} V_k, $$
where $V_k$ are the closures of eigenspaces of the Lie derivative 
$\partial_{e_0}$ for the eigenvalue $ik$, $k\in \Z$. This decomposition 
is preserved by $\tilde D$, since it commutes with the (isometric) 
$U(1)$-action on $M$.
 
Next, let $L:=M\times_{U(1)}\IC$ be the complex line bundle associated to the 
$U(1)$-bundle $M\to N$. In \cite{AB} it is shown that there is a natural 
homothety of Hilbert spaces (isomorphism if the fibres are of length $\ell =1$) 
$$ Q_k:L^2(\Sigma N \otimes L^{-k}) \to  V_k, $$
which satisfies 
 $$Q_k (\gamma _i \Psi )=\gamma_i Q_k (\Psi), \quad i=1,\dots, n$$
and 
\begin{equation}
\label{covder}
  \nabla_{{e}_i} {Q_k(\Psi)} = 
   Q_k(\nabla_{f_i}\Psi) + {1 \over 4} \sum_{j=1}^n
\left(
    \tilde\Gamma_{i0}^j  - \tilde\Gamma_{ij}^0
\right) 
\ga_0\ga_j {Q_k(\Psi)},
\end{equation}
where $f=(f_1, f_2,\dots , f_n) $, $f_i:= \pi_*(e_i)$ is 
a local orthonormal frame on $N$.

Then  $D_h:L^2(\Sigma M) \to L^2(\Sigma M)$ 
is defined as the unique closed linear operator, such that on each $V_k$ 
it is:  
$$ D_h:= Q_k \circ D'_k \circ {Q_k}^{-1},$$
where 
$D'_k=\sum_{i=1}^n \gamma _i \otimes \id\, (\nabla^N_{f_i}\otimes \id+\id\otimes k\nabla^\omega_{f_i})$ 
is the twisted (of charge $k$) Dirac operator on 
$\Sigma N \otimes L^{-k}$.

Here, $\nabla^N$ is the covariant 
spinor derivative on $N$ coming from the Levi-Civita connection on 
$N$, whose Christoffel symbols with respect to the projected frame 
$f=(f_1,\dots, f_n)$ on $N$ are given by
\begin{equation}
\label{GammaN}
\Gamma_{ij}^k=\tilde \Gamma_{ij}^k, 
\qquad \forall \, i,j,k \in \{1,\ldots,n\}
\end{equation}
and $k\nabla^\omega$ is the covariant derivative on 	$ L^{-k}$ of the connection $\omega$.
Using the above results the Dirac operator $\ti D $ on $M$  
can be expressed as a sum
$$ \ti D   = D_v + D_h +   Z, $$
where 
$$ Z:=-(\ell /4)\,\ga_0\,\,\sum_{j<k} d\om(e_j,e_k)\ga_j\ga_k. $$

Observe that since $D_h$, $\gamma_0$ and $Z$ are $U(1)$-invariant, they 
commute with $\partial_{e_0}$. Since for even $n$, $\gamma_1 \gamma_2 \dots \gamma_n$  
anticommutes with any twisted Dirac operator on $N$ and 
$\gamma_0 \sim \gamma_1 \gamma_2 \dots \gamma_n$ (up to a constant $1,i,-1,-i$ depending 
on $n$ and the representation of gamma matrices), $\gamma_0$ anticommutes with $D_h$.

Finally, let us observe that the presence of the zero-order term $Z$ is responsible for 
the torsion-free condition. 
In other words, omitting $Z$ still provides a Dirac operator of $M$
for the linear connection, which preserves the metric $\tilde g$ but has a
nonvanishing (in general) torsion. This can be see easily by looking at the
Christoffel symbols defined by (\ref{GammaN}) and by (\ref{Gamma0}). 
If, in the latter formula we put $\tilde\Gamma_{ij}^k=0$ whenever one or more of the 
indices $i,j,k$ is zero, we get a linear connection, which is still compatible 
with the metric but the components 
\begin{equation}\label{TorsionM}
T_{ij}^0= e^0 (\nabla_{e_i}e_j -\nabla_{e_j}e_i - [e_i, e_j])
 = de^0(e_i, e_j)=\ell \, d\omega (e_i, e_j)
\end{equation}
of the torsion tensor do not vanish (in general).

\section{Noncommutative $U(1)$ principal bundles}

We turn now to the noncommutative picture, where the concept of principal
bundles is given by the Hopf-Galois theory. Let us shortly recall the basic 
definitions, for details and examples see \cite{BM,BM2,DGH,Haj}.

Let $H$ be a unital Hopf algebra and $\CA$ be a right $H$-comodule algebra. 
We will use the natural Sweedler notation for the right coaction of $H$ on $\CA$:
$$ \Delta_R(a) = a_{(0)} \otimes a_{(1)} \in \CA \otimes H.$$ 
We denote by $\CB$ the subalgebra of coinvariant elements of $\CA$.

\begin{dfn}
A $\CB \hookrightarrow \CA$ is Hopf-Galois extension iff the 
canonical map $\chi:$
\begin{equation}
\CA \otimes_\CB \CA \ni  a' \otimes a \mapsto
\chi( a' \otimes a) = a' a_{(0)} \otimes a_{(1)} \in \CA \otimes H,
\end{equation}
is an isomorphism.
\end{dfn}

In the purely algebraic settings the (principal) connections are defined as 
certain maps from the Hopf algebra $H$ to the first order
universal  differential calculus 
$\Omega_u^1(\CA):=\ker m_\CA\subset \CA\otimes\CA$
with the exterior derivative $d_u: \CA\to \Omega_u^1(\CA)$, 
$d_u a := 1\otimes a - a\otimes 1$.

Let $\hbox{Ad}_R(h) = h_{(2)} \otimes S(h_{(1)}) h_{(3)}$ be the right adjoint coaction
and $\varepsilon$ the counit of $H$.
\begin{dfn}
A map $\omega: H \to \Omega^1_u(\CA)$ is 
{\em a strong} universal principal connection if the following conditions hold:
\begin{equation}
\begin{aligned}
& \omega(1) = 0, \\
& \Delta_R  \circ \omega = (\omega \otimes \id) \circ \hbox{Ad}_R \,\, 
{\rm (right\, H\!-\!colinearity)}, \\
& (m \otimes \id) \circ (\id \otimes \Delta_R) \circ \omega = 1 \otimes (\id - \varepsilon)
\,\,{\rm (fundamental\, vector\, condition)},\\
& d_u(a) - a_{(0)} \omega(a_{(1)}) \in \left( \Omega_u^1 (\CB)\right) \CA, \;\;\; \forall a \in \CA
\,\, {\rm (strongness)}.
\end{aligned}\nonumber
\end{equation}
\end{dfn}

As in the case of spectral geometry it will be more convenient to use action of 
the $U(1)$ group rather then the coaction of the algebra of functions over $U(1)$.
Since as the algebra of functions on $U(1)$ we consider the space of polynomials,
and effectively we work with homogeneous elements $a \in \CA^{(k)} \subset \CA$ of 
a fixed  degree $k$, which are defined as follows:
$$ a \in \CA^{(k)} \Leftrightarrow \Delta(a) = a \otimes z^k, $$
we can easily reformulate all conditions above using the language of $U(1)$ 
action, where we have:
$$ a \in \CA^{(k)} \Leftrightarrow  e^{i\phi} \triangleright a = e^{i k \phi} a.$$

\begin{dfn}
\label{strong2} 
For a $U(1)$ Hopf-Galois extension $\CB \hookrightarrow \CA$ we say that 
the map $\omega: C^\infty(U(1)) \to \Omega^1_u(\CA)$ is a {\em strong} universal 
principal connection iff:
\begin{equation}
\begin{aligned}
& \omega(1) = 0, \\
& g \triangleright \omega = \omega, \;\; \forall g \in U(1), \\
& m \circ (\id \otimes \pi_n) \omega(z^k) = \delta_{kn} - \delta_{n0}, \;\;\; \forall k,n \in \Z, \\
& d_u (a) - a \omega(z^k) \in \left( \Omega^1 (\CB)\right) \CA, \;\;\; \forall a \in \CA^{(k)}.
\end{aligned}
\end{equation}
\end{dfn}

Here $\pi_n$ projects an element on the part of a fixed homogeneity degree $n$. 

It is possible to extend this definition of connections for nonuniversal 
differential calculi, however only after requiring certain compatibility 
conditions between the differential calculus on $\CA$ and a given calculus 
over the Hopf algebra $H$. Choosing a subbimodule ${\cal N} \subset \CA \otimes \CA$
we have an associated first order differential calculus over $\CA$. If the 
canonical map $\chi$ maps ${\cal N}$ to $\CA \otimes Q$, where 
$Q \subset \ker \varepsilon \subset H$ is an $Ad$-invariant vector space 
then it is possible to use a calculus over $H$ determined by $Q$ using the
Woronowicz construction of bicovariant calculi \cite{Wor}. For details see 
\cite{BM,Haj,Haj2}.

In what follows we shall need the fact (see \cite{HM}, p.251) that the 
in the case of Hopf-Galois extensions the multiplication map gives 
the natural isomorphism:
\begin{equation}\label{isom}
\left( \Omega^1 (\CB)\right)\otimes_\CB \CA \approx \left( \Omega^1 (\CB)\right) \CA \ .
\end{equation} 

\section{Spectral triples over noncommutative $U(1)$ bundles}

We assume that there exists a real spectral triple over $\CA$ (for details 
on real spectral triples, notation and basic properties we refer to the 
textbook \cite{BFV}), which is $U(1)$ equivariant, that is the action 
of $U(1)$ extends to the Hilbert space and the representation, the Dirac 
operator and the reality structure 
are $U(1)$ equivariant. We denote by $\pi$ the representation of $\CA$ on $\CH$, 
$D$ is the Dirac operator and $J$ the reality structure. 

Let $\delta$ be the operator on $\CH$ which generates the action of $U(1)$ on the
Hilbert space. The $U(1)$ equivariance of the reality structure and $D$ means that:
\begin{equation}
J \delta = - \delta J, \;\;\;
D \delta = \delta D,
\label{reality}
\end{equation}
whereas the equivariance of the representation is:
$$  [\delta, \pi(a)] = \pi(\delta(a)), \;\; \forall a \in \CA, $$
where $\delta(a)$ is the derivation of $a$ arising from the $U(1)$ action.

For simplicity, we take the dimension of the spectral triple over $\CA$ to be 
odd, then the dimension of spectral triple over $\CB$ is even 
(in particular, the spectral triple over $\CB$ has a  $\Z_2$ grading).
We shall require that the signs $\epsilon,\epsilon^\prime$ present in
$J^2 = \epsilon$ and $JD=\epsilon^\prime DJ$, are not changed when 
we pass to the quotient. 
This corresponds, specifically, to the case of the top dimension $3$ and 
the dimension of the quotient $2$. Our example will be therefore 
three-dimensional, we have:
$$ DJ=JD, \;\;\; J^2=-1. $$

Let us finally recall 
that $J$ satisfies $\left[ \pi(a), J\pi(b^*) J^{-1} \right] = 0$ and the order one condition
\begin{equation}
\left[ [D, \pi(a)], J\pi(b^*) J^{-1} \right] = 0, \;\;\; \forall a,b \in \CA, 
\label{ooc}
\end{equation}
which will be important in further considerations.\\

Throughout the rest of the paper we omit writing $\pi$, whenever it is clear from
the context that we mean the image of $a \in \CA$ by the chosen representation $\pi$.
 
\subsection{Projectable spectral triples}

We start by assuming the existence of an additional structure on the 
spectral triple. 

\begin{dfn}
\label{def41}
We say that the $U(1)$ equivariant spectral triple $(\CA,D,J,\CH,\delta)$ is 
{\em projectable} 
along the fibres if there exists an operator $\Gamma$, a $\Z_2$ grading of the Hilbert 
space $\CH$, which satisfies the following conditions:
\begin{equation}\label{projectable}
\begin{aligned}
& \forall a \in \CA: [\Gamma, a]=0, \\
&\Gamma J = - J \Gamma, \;\;\; \Gamma \delta = \delta \Gamma, \;\;\; \Gamma^* = \Gamma, \;\;\; 
\Gamma^2=\hbox{id},
\end{aligned}
\end{equation}
and define the {\em horizontal Dirac operator} as
$$ D_h = \frac{1}{2} \Gamma [\Gamma, D]. $$
\end{dfn}

\begin{rem}
Note that the operator $ D_h$ anticommutes with $ \Gamma $. It will be 
employed for construction of even spectral triples over $\CB$. Note also 
that the signs in the definition are adjusted to the case of dimension 
$3$ bundle over a $2$-dimensional space.
\end{rem}

Although classically this follows, 
we shall impose now certain requirements in order to guarantee that 
the differential calculus over $\CB$ does not depend on the choice 
of projection.

\begin{dfn}
\label{def41a}
We say that the differential calculus $\Omega^1_{D}(\CB)$ on $\CA$ is 
{\em projectable} 
along the fibres if 
\begin{equation}
[D_h, b] = [D, b], \;\; \forall b \in \CB.
\label{prop-b}
\end{equation}
\end{dfn}
This property tells us that the operators representing the 
one-forms over $\CB$ generated by $D_h$  and by $D$ are the same.
Moreover, it follows then that the restriction of 
$[D, b]$ to $\CH_0$ is just equal to the operator $[D_0, b]$ on  $\CH_0$.

Next, we shall make one more additional assumption, which in the classical
case amounts to the situation when the $U(1)$ fibres are of equal length. 
What we propose, is a geometric characterization of the Dirac 
operator, which closely follows the analysis of Amman and B\"ar \cite{AB}. 

\begin{dfn}\label{equall}
We say that the $U(1)$ bundle has fibres of constant length 
(taken to be $2\pi \ell$) if $D_v$ (called the vertical part of the Dirac operator),
$$ D_v = \frac{1}{\ell} \Gamma \delta, $$ 
is such that:
$$ Z = D - D_h - D_v, $$
is a bounded operator such that $Z$ commutes with the elements from the commutant:
\begin{equation}\label{bicomm} 
[Z, J a^* J^{-1}] = 0, \;\; \forall a \in \CA.
\end{equation}
\end{dfn}

\begin{rem}
Note that $Z$ has to commute with $\Gamma$ and also  with $\CB$, due to the 
condition (\ref{bicomm}) and \eqref{prop-b}. 
We observe also that the 
requirement of the boundedness of $Z$ can be reinforced by requiring that $Z$ is
operator of zero order in the sense of generalized pseudodifferential 
operators of Connes and Moscovici \cite{CoMo}.
\end{rem}


We define the space $\CH_k \subset \CH$, $k\in\Z$, to be a subspace of vectors 
homogeneous of degree $k$ in $\CH$ that is, they are eigenvectors of $\delta$ of 
eigenvalue $k$. 
The relation (\ref{reality}) means that 
$$ J \CH_k = \CH_{-k}.$$
In particular the subspace $\CH_0$ is $J$ invariant. From the equivariance of $D$
we see that each $\CH_k$ is preserved by the action of $D$: 
$$D \CH_k \subset \CH_k.$$
Since $[\delta, D_h] = 0$ we see that $D_h$ preserves each of the subspaces $\CH_k$. 
We shall denote by $D_k$ its restriction to each subspace $\CH_k$. 
Similarly, we denote by $\gamma_k$ the restriction of $\Gamma$ to $\CH_k$ 
and by $j_k$ the restriction of $J$ ($j_k$ is a map $\CH_k \to \CH_{-k}$). 

Now we have:
\begin{pro}\label{Dk}
The operators $D_k, \gamma_k,  j_k$ satisfy the commutation relations
\begin{equation}
\gamma_k D_k = - D_k \gamma_k,\quad
j_k D_k = D_{-k} j_k, \quad
j_k \gamma_k = - \gamma_{-k} j_k.
\label{Dhrel}
\end{equation}
The data $(\CB, \CH_0, D_0, \gamma_0, j_0)$ give an even real spectral triple 
of KR-dimension $2$ over $\CB$. For $k \not= 0$, $(\CB, \CH_k, D_k, \gamma_k)$ 
are even spectral triples, which are pairwise real.
\end{pro}

\begin{proof}
Clearly $D_h$ is a selfadjoint operator, which has the same commutation
relation with $J$ and $\Gamma$ as $D$ \footnote{In the case of dimension other than $3$
it is possible to adjust the signs in the definition of $\Gamma$ and $j_k$
so that the resulting $KR$-dimension of the projected spectral triple shall
be correct.}. 
Therefore, the relations (\ref{Dhrel}) follow after restriction 
to subspaces. 

That each $D_k$ has bounded commutator with the elements from $\CB$ 
is an immediate consequence of the fact that it is a restriction of 
$D_h$, which has this property. To see that $D_k$ has a compact 
resolvent observe that $D-Z$, which 
is a bounded perturbation of $D$, is again $U(1)$ invariant.
and  we can restrict it to $\CH_k$. Its eigenvalues are: 
$$ \pm \sqrt{k^2/\ell^2 + \lambda^2_{(k)}} $$ 
where $\lambda_{(k)}$ are eigenvalues of $D_k$.
Since $D$ has compact resolvent, $|D_k|$ must have pure spectrum diverging to $\infty$ 
and so $D_k$ also has a compact resolvent.

The KR-dimension follows from the relations (\ref{Dhrel}) and $J^2=-1$.
\end{proof}

Now, few remarks are in order.
\begin{rem}
Note that taking a pair $\CH_k \oplus \CH_{-k}$ yields again 
a real even spectral triple, which is, however, reducible (in the sense of \cite{ISS}). 
We shall see in the next section that each triple built on a single $\CH_k$
 is in fact a spectral triple twisted by a module over $\CB$.
\end{rem}

\begin{rem}
As follows from the proof of Proposition \ref{Dk}
the asymptotic spectral properties of $D_k$ are 
the same as properties of $D$ restricted to $\CH_k$. 
Therefore spectral dimension of each $D_k$ can be at most the same as that of $D$, 
which however does not imply that it is exactly $2$ as we know in the classical case.
\end{rem}

\begin{rem}
\label{fiblen}
In the classical situation, when one is able to consider the fibres 
over points of the base space, there is no problem to define the length 
of a fibre and, consequently, to restrict the considerations to the case 
when all fibres are of equal length.
In the general noncommutative setup, this is no longer possible. 
Instead, we have proposed in the definition \ref{equall} above
how to replace and extend this property in a way which links 
the length of fibres to the form of the Dirac operator.
There may be, however, some other alternatives. We mention here just one other 
possible definition, which will work in the case of finitely summable spectral 
triples of metric dimension $n$.

We can say that the $U(1)$ bundle $\CA$ with the equivariant spectral triple 
and the Dirac $D$ has fibres of length $\ell$, if the restriction of $D$
to a $U(1)$ invariant subspace of $\CH$ is an operator of spectral dimension
$(n-1)$ and for any element $b \in \CB$:
$$ \ncint b |D|^{-n} =  \ell \ncint b |D_0|^{-n+1}, $$
where $D_0$ is the restriction of $D$ to the invariant Hilbert subspace $\CH_0$.
The symbol $\ncint$ denotes the noncommutative integral:
$$
\ncint T:= \underset{s=0}{\hbox{Res}}\ \zeta_D^{T}(s)
=\underset{s=0}{\hbox{Res}}\ \hbox{Tr}\,\big(T|D|^{-s}\big).
$$ 
Although in the classical case this definition implies that the fibres are 
indeed of equal length, in the noncommutative case it is far from being clear. 
The advantage of this definition, is that it is not sensitive
to the bounded perturbations of the Dirac operator, which do not commute with the 
algebra elements. 
\end{rem}

From now on, for simplicity we assume $\ell = 1$.

\subsection{Spectral triples twisted by a module}

In this subsection we shall discuss how to twist real spectral 
triples by a left-module  equipped with a hermitian connection. 

Instead of using the usual algebraic definition of connection, we shall use 
the operator language suited best for the further considerations. 
Let $(\CB, \CH, D, J)$ be a real spectral triple over an algebra $\CB$.
Let $\CH_M$ be another Hilbert space with a representation of $\CB$. 
Let $M$ be space of (some) $\CB$-linear maps $\CH \to \CH_M$ such that\\
({\em i}): $ M(\CH) =\CH M$ is dense in $\CH_M$ and\\ 
({\em ii}): the 
multiplication map from $\CH \otimes_\CB M$ to $\CH M$ is an isomorphism.
 
Using the right $\CB$-module structure on $\CH$, 
$ hb = J b^* J^{-1} h, \;\; \forall b \in \CB, h \in \CH$,
we introduce a left $\CB$-module structure on $M$ through:
$$ (b m)(h) = m(hb), \;\;\; \forall m\in M.$$
We find it convenient to write the action of $m$ on $h$ from the right,
i.e. $h m$ instead of $m(h)$, then the left $\CB$-linearity of $m\in M$ 
reads 
$$(b h) m = b (h m),$$
while the left $\CB$ action on $M$ becomes 
$$ h(bm) = (h b) m, \;\;\; \forall h \in \CH, b \in \CB.$$

It follows from the order one condition that there is 
a right action of $\Omega^1_D(\CB)$ on $\CH$, 

\begin{equation}
 h \omega := -J \omega^* J^{-1} h, \label{fomuri}
\end{equation} 
which is left $\CB$-linear, where $\omega^*$ is the 
adjoint of the operator $\omega$ so that:
$$ ([D, b])^* = - [D, b^*],$$
and
$$ h ([D, b]) =  D(hb) - (Dh) b, \;\;\; \forall h \in \CH, b\in \CB. $$

This induces a left action of $\Omega^1_D(\CB) $ on $M$,
and $\Omega^1_D(\CB)\bullet M$ is just the space of all compositions $m \circ \omega$ 
of left $\CB	$-linear maps: (right action of) $\omega \in \Omega^1_D(\CB)$ 
and $m \in M$. 
Note that the action of such a composition on $h$, when written from the right, becomes
$$ (m \circ \omega ) (h) = h (\omega m) = (h\omega) m. $$

It is easy to verify that the composition is compatible with the right $\CB$ action on 
$\Omega^1_D(\CB)$ and the left $\CB$ action on $M$
$$ h (\omega b) m = h \omega (b m), \;\;\; \forall b \in \CB, m\in M, h \in \CH.$$

Next we pass to connections (covariant derivatives) on $M$.
Usually they would take values in $\Omega^1_D(\CB)\otimes M$,
but inspired by \eqref{isom} with its right side regarded as operators 
we propose the following definition.

\begin{dfn}
\label{Dconndef}
We call a linear map $\nabla: M \to \Omega^1_D(\CB)\bullet  M$ 
a {\em $D$-connection} on $M$, if  it satisfies the Leibniz rule
$$ \nabla(b m) = [D,b] m + b \nabla(m), \;\;\; 
\forall b \in \CB, m \in M. $$
We say that the connection is {\em hermitian} if for each
$m_1, m_2\in M$
\begin{itemize}
\item 
as an operator on $\CH$, $m_1^\dagger \circ m_2 \in J \CB J^{-1}$, 
(so it is in the commutant of $\CB$);
\item
writing the actions on arbitrary $h \in \CH$ from the right:
\begin{equation}
h \nabla(m_2) m_1^\dagger - h m_2 \nabla(m_1)^\dagger =   
(D h) m_2 m_1^\dagger  - D(h m_2 m_1^\dagger). \label{mherm} 
\end{equation}
\end{itemize}
\end{dfn}

Note that $D$-connection (which could be also called a covariant 
$D$-derivative) over the module $M$ uses the differential calculus, 
which comes from the spectral 
triple. Note also that due to the first condition $M$ cannot consist of all 
possible $\CB$-linear maps, and in particular they have to preserve 
the domain of $D$.\\

Finally, using a hermitian D-connection and the properties of $M$ we define 
an (unbounded) twisted Dirac operator $D_M$.

\begin{dfn}\label{defDM}
Define an operator $D_M$ over the the dense domain in $\H_M$ by 
\begin{equation}
 D_M (h m) = (D h) m + h \nabla(m), \;\;\; \forall h\in \CH, m \in M.
 \label{twistedD}
\end{equation}
\end{dfn}
The denseness follows from the assumption i) on $M$, while to see that $D_M$ is
well defined, i.e. it depends only on the product $hm$,
due to the assumption ii) it suffices to check that for $b\in \CB$
$$
(D(hb)) m + hb \nabla(m) = (D h) bm + h \nabla(bm),
$$
which follows from the Leibniz rule for $\nabla$.

\begin{pro}
\label{Dconnpro}
The operator $D_M$ is selfadjoint, has bounded commutators with $\CB$ 
and compact resolvent.
\end{pro}
\begin{proof}
We show that $D_M$ is symmetric
$$
\begin{aligned}
\left( h_1m_1, D_M (h_2 m_2) \right) 
&=   \left( h_1 m_1, (D h_2) m_2 \right) 
   + \left( h_1 m_1, h_2 \nabla(m_2) \right) \\
&=   \left( h_1, (D h_2) m_2 m_1^\dagger \right) 
   + \left( h_1, h_2 \nabla(m_2) m_1^\dagger \right) \\ 
&=   
\left( h_1, D (h_2 m_2 m_1^\dagger) \right) 
   + \left( h_1, h_2 m_2 \nabla(m_1)^\dagger \right)\\
& = \left( (D h_1) m_1, h_2 m_2 \right) 
  + \left( h_1 \nabla(m_1), h_2 m_2 \right) \\
& = \left( D_M(h_1 m_1), h_2 m_2 \right),
\end{aligned}
$$ 
where in the third equality we used (\ref{mherm}).
However, as $m \in M$ and $\nabla(m)$ are all bounded operators 
and $D$ is selfadjoint then it is clear that $D_M$ is selfadjoint as well.

Next we compute the commutator with $b \in \CB$,
$$
 [ D_M, b] (h m) = D_M (b h m) - b D_M (h m) 
$$
$$
 = (D b h) m + (b h) \nabla(m) - b (D h) m - b (h \nabla(m)) \\
 = ([D,b] h) m, 
$$
and hence $||[D_M, b] || \leq ||[D,b] ||$.

We pass to show that $D_M$ has compact resolvent.
When $M$ is a finite free module over $\CB$ with the basis  $m_i$ over $B$
we have (summation over $j$ implied)
$$ D_M (h_j m_j) = (D h_j) m_j + h_j \nabla(m_j). $$
Now, $\nabla(m_j)$ can be written as $\omega_{j,k} m_k$ 
and hence the second part of the expression is in fact a bounded operator 
on $\CH_M$. Therefore, the spectrum of $D_M$ is at most a spectrum of
a bounded perturbation of $D$, which has a compact resolvent, so must have $D_M$.
Similar discussion applies to the case when $M$ is a finite projective module 
over $\CB$.
\end{proof}

\section{Connections}

In order to define a strong $U(1)$ connection over the bundle using the differential calculus 
given by the Dirac operator, we need to impose certain conditions on the Dirac operator itself. 

\begin{dfn}
\label{compcal}
We say that the first order differential calculus $\Omega^1_D(\CA)$, over $\CA$ 
given by the Dirac operator $D$ is {\em compatible} with the standard de Rham 
calculus over $U(1)$ if the following holds:
\begin{equation}
\forall p_i,q_i \in \CA: \sum_i p_i [D, q_i] = 0 \Rightarrow \sum_i p_i \delta(q_i) = 0.
\label{eq43}
\end{equation}
\label{def51}
\end{dfn}
Recalling the standard characterization of the first order 
de Rham calculus over an algebra of functions on a group via
$Q=(\ker  \varepsilon)^2$, the compatibility as discussed 
at the end of Sect.3 follows from the next lemma.

Let the first order differential calculus over $\CA$ be given by 
a subbimodule ${\cal N}$: 
$$ {\cal N} \subset \ker  m \subset \CA \otimes \CA, $$ 
defined by the relation
$$ \sum_i p_i \otimes q_i \in {\cal N} \; \Leftrightarrow \; \sum_i p_i [D, q_i] = 0.$$ 

\begin{lemma}
The image of ${\cal N}$ by the canonical map $\chi$  
is in $\CA \otimes (\ker \varepsilon)^2$.
\end{lemma}

\begin{proof}
Let $ \sum_i p_i \otimes q_i \in {\cal N}$. 
Decompose $q_i$ as a sum of homogeneous elements of $\delta$:
$$ q_i = \sum_n q_i^{(n)}. $$
Then since $ \sum_i p_i \otimes q_i \in \ker  m$, 
using the identity (\ref{eq43}) we have,
$$ \sum_{i,n} p_i q_i^{(n)} = 0, \quad \sum_{i,n} n p_i q_i^{(n)} = 0. $$
which we can solve for $p_i q_i^{(0)}$ and $p_i q_i^{(1)}$:
$$ 
\begin{aligned}
& \sum_i p_i q_i^{(1)} = - \sum_{i,n \not=1} n p_i q_i^{(n)}, \\
& \sum_i p_i q_i^{(0)} = - \sum_{i,n \not=0} p_i q_i^{(n)}
= \sum_{i,n \not=0,1} (n-1) p_i q_i^{(n)}.
\end{aligned}
$$

Applying canonical map $\chi$ to $\sum p_i \otimes q_i$ we obtain:

$$ 
\begin{aligned}
\chi(\sum_i p_i \otimes q_i) & = \sum_{i,n} p_i q_i^{(n)} \otimes z^n \\
&= \sum_{i,n \not=0,1} p_i q_i^{(n)} \otimes (z^n - 1 -n(z-1)).
\end{aligned}
$$
The second factor on the right-hand side is in $(\ker  \varepsilon)^2$ 
since it can be written as 
$$(z-1)(z^{n-1}+\dots +z+1-n).$$ 
\end{proof}

Now, we are ready to define a strong connection for a principal $U(1)$ 
bundle with a differential calculus set by the Dirac operator.

\begin{dfn}
\label{defcon}
Let $\Omega^1_D(\CA)$ be as in  definition \ref{compcal}.
We say that $\omega \in \Omega^1_D(\CA)$ is a strong principal connection 
for the $U(1)$ bundle $\CB \hookrightarrow \CA$ if the following conditions hold:
\begin{align*}
& [\delta, \omega] = 0, \;\; \hbox{\em ($U(1)$ invariance of $\omega$)} \\
& \hbox{if\ } \omega = \sum_i p_i [D, q_i]\; \hbox{then} \sum_i p_i \delta(q_i) = 1, \;\; 
\hbox{\em (vertical field condition)}, \\
& \forall a \in \CA: [D,a] - \delta(a) \omega \in \Omega_D^1(\CB) \CA, 
\;\; \hbox{\em (strongness)} 
\end{align*}
\end{dfn}

Observe that the second condition (which in the classical case corresponds to 
the value of $\omega$ on fundamental vertical vector field) makes sense 
due to assumption (\ref{eq43}). Our definition is motivated by the classical 
example of adapting the definition \ref{strong2} 
to the case of the de Rham calculus over a $U(1)$ bundle. 
In such a situation one has $\omega(z^n) = n \omega(z)$ 
for the functions on $U(1)$, therefore the strong principal connection is completely
defined by $\omega=\omega(z)$.

We shall see in section 5 that the fourth condition (strongness) will play a significant role in 
the extension of Dirac operator. 


\subsection{Lifting the Dirac operator through connection}

In this section we shall now put together the construction from
the section 4.2 and the notion of strong connections from last section
to obtain the construction of principal connections compatible with 
spectral triples.
The role of the Hilbert space $\CH$, the Hilbert space $\CH_M$ 
and the (sub)space $M$ of operators from $\CH$ to $\CH_M$
(all of them left $\CB$-modules) will be played by,  respectively, 
the Hilbert space $\CH_0$,  the Hilbert space $\CH_k$ and $\CA^{(k)}$ 
(the eigenspaces of $\delta$). 

Here $\CA^{(k)}$ is assumed to satisfy the properties ({\em i}) and ({\em ii})
required for $M$, and it acts on $\CH_0$ from the right, or more precisely 
$M=J\CA^{(k)*}J^{-1}$. 

\begin{pro}
Let $(\CA,D,\CH, J)$ be the projectable real spectral triple of a $U(1)$ 
bundle with a projectable differential calculus $\Omega^1_{D}(\CB)$
as defined in section 4, and let $\omega$ be the strong connection 
in the sense of definition (\ref{defcon}). Then the map:
$$ \nabla_\omega: \CA^{(k)} \ni a \mapsto \nabla_\omega (a):=[D,a] - k a \omega ,$$ 
where both $a$ 
in the domain of $\nabla_\omega$, and its image $\nabla_\omega (a)$
are regarded as operators on $\CH_0$ acting from the right,
defines a $ D_0$-connection (as defined in (\ref{Dconndef}))
over left $\CB$-module $\CA^{(k)}$.
\end{pro}
\begin{proof}
First, observe that (thanks to the real structure)
the elements of $\CA^{(k)}$ indeed map $\CH_0$ to $\CH_k$ as $\CB$-linear 
morphism.

To see that $\nabla_\omega$ is well-defined observe first that due to the 
strongness condition and the fact that $\omega$ is $U(1)$-invariant the 
element $\nabla_\omega(a)$ is indeed in the restriction to $\CH_0$ of 
$\Omega^1_D(\CB)\bullet \CA^{(k)}$.
But due to the assumption that $[D,b]=[D_h,b]$ for $b\in \CB$,
the restriction of $\Omega^1_D(\CB)$ to $\CH_0$ is just $\Omega^1_{D_0}(\CB)$.
Hence we see that $\nabla_\omega(a)$ 
(regarded as an operator acting on $\CH_0$ from the right)
indeed belongs to $\Omega^1_D(\CB)\bullet \CA^{(k)}$
(in spite we used the commutator with $D$ in its definition).
Moreover $\nabla_\omega$ satisfies the Leibniz rule:
$$
\nabla_\omega (ba) = [D,ba] - kba\omega   = [D_0,b]a + b\nabla_\omega  (a).
$$
This shows that $\nabla_\omega$ is a $D_0$-connection. 
\end{proof}


Further, we have:
\begin{lemma}
The $D_0$-connection $\nabla_\omega$ is hermitian if $\omega$ is 
selfadjoint (as an operator on $\CH$).
\end{lemma}

\begin{proof}
We compute (\ref{mherm}) for $a_1,a_2 \in \CA^{(k)}$, $h \in \CH_0$,
using the rule (\ref{fomuri}):
$$
\begin{aligned}
h \left( \nabla(a_2) a_1^\dagger \right. - & \left. a_2 \nabla(a_1)^\dagger 
- (Dh) a_2 a_1^\dagger + D( h a_2 a_1^\dagger) \right) \\
& = h \left( [D,a_2] - k a_2 \omega \right) a_1^\dagger 
    - h \left( a_2 ([D,a_1] - k a_1 \omega)^\dagger \right) 
    - h [D, a_2 a_1^\dagger] \\
& = h \left( k a_2 (\omega^\dagger - \omega) a_1^\dagger \right).  
   \end{aligned}
$$
\end{proof}
\begin{dfn}
Let $\omega$ be a principal connection and  $\nabla_\omega$ the related $D_0$-connection
on $\CA^{(k)}$,  $k\! \in\! \Z$. We construct as in section 4.2 the twisted spectral triple 
$(\CB, \CH_k, D_\omega^{(k)})$, where the twisted Dirac operator $D_\omega^{(k)}$ on $\H_k$
is the closure of $(D_0)_{\CA^{(k)}}$ (see \eqref{twistedD})
defined on a dense subset $\hbox{Dom}(D) \CA^{(k)}$ of $\CH_k$. Taking $D_\omega$ to be 
the closure  of the direct sum of the operators $D_\omega^{(k)}$ for all $k \in \Z$, 
we obtain the twisted Dirac operator $D_\omega$ on $\CH$. 
\end{dfn}

Note that notwithstanding Proposition \ref{Dconnpro} valid for $D_\omega^{(k)}$,
from the construction it is not entirely obvious that the Dirac operator $D_\omega$ 
is selfadjoint or has bounded commutators with the elements from the algebra $\CA$ 
Therefore we have:
\begin{pro}
\label{mpro}
The twisted Dirac operator $D_\omega$ is selfadjoint if $\omega$ is a selfadjoint 
one form and has bounded commutators with all elements of $\CA$. 
\end{pro}

\begin{proof}
We compute the action of $D_\omega$ on an element $h p$
in its domain, with $h \in \CH_0$ and $p \in  \CA^{(n)}$:
$$
\begin{aligned}
D_\omega (hp) =& (D_0 h) p + h [D,p] - n h p \omega \\
= & (D_0 h) p + [D, J p^* J^{-1}] h  + J \omega^* J^{-1} n hp \\
= & D (hp) + \left( (D_0-D) h \right) p + J \omega^* J^{-1} \delta( hp) \\
= & \left( D + J \omega^* J^{-1} \delta - Z \right) (hp),
\end{aligned}
$$
where we have used the decomposition of $D$ into the horizontal part 
(which restricted to $\CH_0$ is $D_0$), the vertical part (which vanishes 
on $\CH_0$) and the bounded perturbation $Z$, which by \eqref{bicomm} 
obeys $(Zh)p=Z(hp)$. 

Hence we see that $D_\omega$ on $\CH_n$ is a reduction of the 
operator $D + J \omega^* J^{-1} \delta + Z$ to this subspace, so the
latter equals the direct sum of the reductions. Now, $D$ and $\delta$ 
are selfadjoint by definition, $Z$ and $\omega$ are by assumption 
bounded and selfadjoint. Since $\delta$ is relatively bounded 
with respect to $D$, by Kato-Rellich theorem $D_\omega$ is selfadjoint on 
$\CH$.

To see that the commutators with the elements from $\CA$ 
are bounded we observe that $D$ has bounded commutators with
each $a \in \CA$ and since $\omega$ is a one-form, from the 
order one condition the commutator of the second term with 
$a$ is $ J \omega^* J^{-1} \delta(a)$ and hence it is bounded. 
The third term, as a commutator of two bounded elements 
remains bounded. Hence, $[D_\omega,a]$ is bounded for any $a$.
\end{proof}

Observe that strongness is necessary to define a $D_0$ connection and then, 
in turn, to twist the Dirac operator by $\omega$ on every module $\CA^{(k)}$. 
On the other hand, the vertical field condition assures that the $D_\omega$ 
is indeed a horizontal part of some Dirac operator on $\CA$. However, for that 
we shall need a stronger condition on the bounded operator $Z$ than its 
commutativity with $\CB$ (and with  with $\Gamma$), namely that it commutes 
with $\CA$.

\begin{pro}
Let  $Z$ be as in sect. 4.1. Let 
\begin{equation}\label{compatib}
{\cal D}_\omega = \Gamma \delta + D_\omega .
\end{equation}
If $Z$ commutes with the algebra $\CA$ then $(\CA, {\cal D}_\omega, \CH)$ 
is a projectable spectral triple with equal length fibres 
and the horizontal part of the operator ${\cal D}_\omega$ coincides 
with $D_\omega$. 
\end{pro}

\begin{proof}
The operator ${\cal D}_\omega$ can be written as:
$$ {\cal D}_\omega  = D + (J \omega J^{-1} + \Gamma)\delta + Z.$$
Then using similar arguments as in the case of $D_\omega$ we see that it
is selfadjoint and has bounded commutators with the elements of $\CA$.
 
We show now that ${\cal D_\omega}$ has a compact resolvent. Since $D_\omega$ 
is a twist of $D_0$ by a connection, then by the arguments from section 4.2 it 
has a compact resolvent when restricted to each $\CH^{(n)}$. Since $D_\omega$ 
then anticommutes with $\Gamma$, squaring ${\cal D}_\omega$ we obtain:
$$ {\cal D}_\omega^2 = D_\omega^2 + \delta^2. $$
Now, it is easy to check that the spectrum of this operator consists of
the values of the form $n^2 + \lambda^2_n$, where $\lambda_n$ is in the
spectrum of $D_\omega$, restricted to $\CH^{(n)}$. Since for each $n$
the sequence $\{ \lambda_n \}$ cannot have an accumulation point other than 
infinity, we see that it must be true also for $\{ n^2 +\lambda^2_n \}$.

This shows that $(\CA, {\cal D}_\omega, \CH)$ is a spectral triple 
(not necessarily real). It is clearly projectable with the original 
grading $\Gamma$. \\

We compute now the horizontal part of ${\cal D}_\omega$:
\begin{equation}
\begin{aligned}
({\cal D}_\omega)_h h =& \frac{1}{2} \Gamma [\Gamma, {\cal D}_\omega] h =
                         \frac{1}{2} \Gamma [\Gamma, D_\omega] h \\
 =& \frac{1}{2} \Gamma \left[ \Gamma, \left( 
 D + J {\omega} J^{-1} \delta  - Z \right) \right] h \\
= & \left(  D_h +  
  J  \left( \frac{1}{2} \Gamma \left[ \Gamma, {\omega}  \right] \right) J^{-1} \delta \right) h \\
= & \left(   D_h  + J \omega J^{-1} \delta + \Gamma \delta \right) h \\
= & D_\omega h.
\end{aligned}
\label{vfc2}
\end{equation}
Here we used that for a general $\omega = \sum a_{i} [D, a_i']$, we have:
$$ \omega = \sum a_{i} \left( [D_h, a_i'] + [\Gamma \delta, a_i'] + [Z, a_i'] \right). $$
The first term anticommutes with $\Gamma$ whereas the other two commute by 
construction. If, as assumed $Z$ commutes with the elements of the algebra then
$$ \omega = \sum a_{i} [D_h, a_i'] + \sum \Gamma a_i \delta (a_i') $$
and using the vertical field condition we have:
$$ \omega = \frac{1}{2} \Gamma \left[ \Gamma, {\omega}  \right] + \Gamma. $$
\end{proof}

Observe that if we do not assume that $Z$ commutes with $\CA$ (note that this together 
with (\ref{bicomm}) is a severe limitation for $Z$) we cannot have the
property $({\cal D}_\omega)_h = D_h$. Nevertheless, still the spectral
triple will be projectable, the restriction of both operators to $\CH_0$ and 
the resulting spectral triple over $\CB$ will be identical. However, without
that assumption we cannot obtain a spectral triple with fibres of equal length.

We recapitulate our constructions so far as the following sequence of steps: 
start with a full Dirac operator of a projectable real spectral triple, take a strong 
connection, define a $D$-connection and then twist by it the projected Dirac operator, 
and finally add to it $\Gamma\delta$.
This produces a new Dirac-type operator ${\cal D}_\omega$ over $\CA$. 
Of course, it  is just one of the class of other possible Dirac operators, 
since we may take any bounded perturbation of ${\cal D}_\omega$.
For the setup we have just discussed it is worth to formulate the 
following definition:
\begin{dfn}
\label{compconn}
We say that the connection $\omega$ is {\em compatible} with the Dirac 
operator $D$ if $D_\omega$ and $D_h$ coincide on a dense subset of 
$\CH$.
\end{dfn}

Following results of \cite{AB} it is not difficult to see that the classical
Dirac operator, which arises from the metric compatible with a $U(1)$ principal 
connection (in the sense of the construction in the section 2) is indeed 
compatible with this connection in the sense of the above definition, whereas
the operator ${\cal D}_\omega$ differs from it by a torsion term. 

It is worth to mention that in the genuine noncommutative situation
we don't expect $Z$ to commute with $\CA$, in this case ${\cal D}_\omega$
given by \eqref{compatib} and $D_\omega$ would still coincide up to some 
bounded perturbation. 

\section{The noncommutative torus}

To see how the definitions work in a noncommutative case, we study in 
detail the case of the $3$-dimensional noncommutative torus as a $U(1)$ 
bundle over the $2$-dimensional noncommutative torus. 

We choose the generators of the $\T^3_\theta$ to be $U_i$, $i=1,2,3$, 
with the relations 
$$U_j U_k = e^{2\pi i\theta_{jk}} U_k U_j,$$ 
where  $\theta_{jk}$ is an antisymmetric real matrix. We assume that neither 
of its entries is rational. An element of the algebra of smooth functions 
on the  noncommutative torus is of the form:

$$ a = \sum_{k,l,m \in \Z} \alpha_{klm} U_1^k U_2^l U_3^m, $$
where $\alpha_{k,l,m}$ is a rapidly decreasing sequence. The canonical 
trace on the algebra is:
$$ \tau(a) = \alpha_{000}. $$

We start with the canonical Hilbert space $\CH_0$ of the GNS construction with 
respect to the trace $\tau$ on $\T^3_\theta$, and $\pi$ the associated faithful 
representation. With $e_{k,l,m}$ the orthonormal basis of $\CH_0$ we have:

$$ 
\begin{aligned}
U_1 e_{k,l,m} &= e_{(k+1),l,m}, \\
U_2 e_{k,l,m} &= e^{2\pi ik \theta_{21}} e_{k,(l+1),m}, \\
U_3 e_{k,l,m} &= e^{2\pi i(k \theta_{31}+l\theta_{32})} e_{k,l,(m+1)},
\end{aligned}
$$

where $k,l,m$ are in $\Z$ or $\Z+\oh$ depending on the choice of the spin 
structure. 
The projectable spin structures must have the trivial spin structure on the fibre, so $m \in \Z$,
which we assume from now on.

We double the Hilbert space taking $\CH = \CH_0 \otimes \C^2$, with the diagonal
representation of the algebra. The canonical equivariant spectral triple over 
$\T^3_\theta$ is given by the Dirac operator $D$ and the reality structure $J$ 
of the form:

\begin{equation}
D = \sum_{j=1}^3 \sigma^j \delta_j, \;\;\;\;\; J = i \sigma^2 \circ J_0, 
\label{dirtor}
\end{equation}
where $\sigma^i$ are selfadjoint Pauli matrices \cite{BFV}. 
$J_0$ here is the canonical Tomita-Takesaki antilinear map on 
the Hilbert space $\CH_0$.

$$ J_0 e_{k,l,m} = e_{-k,-l,-m}, $$
and $\delta_i$ are selfadjoint derivations acting diagonally on the basis:
 
$$ 
\delta_1 e_{k,l,m} = k e_{k,l,m}, \;\;\;
\delta_2 e_{k,l,m} = l e_{k,l,m}, \;\;\;
\delta_3 e_{k,l,m} = m e_{k,l,m}.
$$

Observe that $JD = DJ$ as required in the spectral triple of dimension $3$. 
We choose the following $U(1)$ action on the torus $\T^3_\theta$, defined through
the action on the generators:
\begin{equation}
e^{i\phi} \acts U_1 = U_1, \;\;
e^{i\phi} \acts U_2 = U_2, \;\;
e^{i\phi} \acts U_3 = e^{i\phi} U_3, 
\label{u1act}
\end{equation}
and the induced diagonal action on the Hilbert space: 
$$ 
e^{i\phi} \acts e_{k,l,m} = e^{i m \phi} e_{k,l,m},
$$

The $U(1)$ invariant subalgebra is generated by $U_1$ and $U_2$, 
and can be identified with the $2$-dimensional noncommutative torus $\T^2_\theta$,
where  the indices of $\theta_{ij}$ run over $1,2$. It is straightforward to check 
that $\T^3_\theta$ is a Hopf-Galois extension of $\T^2_\theta$ (when $\theta_{12}$
match).

The chosen Dirac operator is one which is fully equivariant, that is invariant 
under all three $U(1)$ symmetries. This is not necessary in our case, as we 
need only one $U(1)$ invariance. In particular, we shall allow for the coupling 
of $D$ to gauge fields (or for the 'internal fluctuations' in terminology 
of  \cite{Co2}).

\begin{rem}
The space of possible fluctuations of the Dirac operator of the gauge
connection type (by one-forms) is given by:
$$ D_{A} = D + \sigma^i A_i + J(\sigma^i A_i)J^{-1}, $$ 
where $A_i \in \T^3_\theta$, $i=1,2,3$, satisfy: $ A_i = A_i^*$.
Since $D_{A}$ is a bounded perturbation of $D$ it is not difficult to see that $\CA, \CH, D_{A}$, 
is a (real) spectral triple.
\end{rem}

In the sequel we further require that the Dirac operator $D_{A}$ is $U(1)$-invariant 
and thus we restrict $A_i$ to belong to the invariant subalgebra 
$\T^2_\theta$.

Before we proceed, observe that any one-form in $\Omega^1_{D_A}$ is 
(trivially) a one-form in $\Omega^1_D$. This provides us with 
a more convenient description of the bimodule of one-forms.
We begin with:

\begin{lemma}
The differential calculus generated by $D_A$ satisfies the 
compatibility condition from the definition \ref{def51} if 
$A_3=0$. 
\end{lemma}

\begin{proof}
Take $p_i,q_i$ such that $\sum_i p_i [D,q_i] = 0$. Since $\sigma^i$ are 
linearly independent, we have for $\sigma^3$:
$$ \sum_i p_i \left( \delta_3(q_i) + [A_3, q_i] \right). $$
If $A_3=0$ then the condition follows. 
\end{proof}

This condition on $A$ agrees with the next Lemma on projectability.
\begin{lemma}
If $A_3=0$ then there exists a unique operator 
$$ \Gamma = \sigma^3,$$ 
which makes the spectral geometry over $\T^3_\theta$ 
projectable (in the sense of definition \ref{def41})  with constant length fibres.
\end{lemma}

\begin{proof}
It is easy to see that indeed $\Gamma=\sigma^3$ does satisfy the projectability and equal length fibers. Conversely, any operator $\Gamma$, which commutes with
the algebra, is $U(1)$-invariant, anticommutes with $J$ and is
a $\Z_2$ grading must be a linear combination of Pauli matrices. 
Then, the requirement \ref{def41} that $D$ has equal length fibres 
fixes $\Gamma$ to be $\sigma^3$.
\end{proof}
Furthermore with we have other essential properties discussed in sect. 4.1.
\begin{lemma}
If $A_3=0$ and $\Gamma=\sigma^3$, then the 
horizontal part
of $D_A$ restricted to $\CH_k$ 
gives an even spectral triple over the two-dimensional torus $\T^2_\theta$,
which is real if $k=0$ and pairwise real if $\pm k\in \N$,
 and the differential calculus $\Omega_{D_A}(\CA)$ is projectable
in the sense of Definition \ref{def41a}.
\end{lemma}
\begin{proof}
It is not difficult to recognize that $(D_A)_h$ restricted to $\CH_k$
is a fluctuated Dirac operator over the two-torus. 
To see that for every $b \in \T^2_\theta$ we 
have on a dense domain in $\CH$:
$$ [D_A,b] = [(D_A)_h,b],$$
it is sufficient to see that $\delta(b)=0$ and since $A_3=0$, only the
horizontal part has non-trivial commutators. Furthermore,
since both $A_1$ and $A_2$ must be $U(1)$ invariant, they are, in fact
elements of the invariant algebra $\T^2_\theta$. Therefore, 
a commutator of $(D_A)_h$ (or $D_A$) with $b \in \T^2_\theta$ is in fact
a one-form, which contains only elements from $\T^2_\theta$. Hence,
the restriction to $\CH_0$ is an isomorphism and therefore the bimodule
of one-forms over $\T^2_\theta$ generated by $(D_A)_0$ is the same as the one
generated by $(D_A)_h$.
\end{proof}

Note that, so far, there is a remarkable consistency in all the conditions that
we are imposing on the Dirac operator and the spectral geometry of the 
noncommutative three-torus, viewed as a $U(1)$ bundle over the noncommutative
two-torus. 
But we find that $D_A$ with $A_3\!=\!0$ 
satisfies also the condition of equal length fibres proposed in the remark \ref{fiblen}.
\begin{lemma}
The Dirac operator $D_A$, with $A_3=0$ satisfies
$$ \ncint b |D_A|^{-3} =  \ncint b |(D_A)_0|^{-2}. $$
\end{lemma}
\begin{proof}
Indeed, using the results of the explicit calculations of the spectral 
action \cite{MCC}, it can be seen that the noncommutative integral of the perturbed Dirac
operator $|D_A|^{-3}$ does not depend on $A$. Hence, we can work with
$D$ alone and its restriction to $\CH_0$. Then, the proof is reduced
to a simple exercise. 
\end{proof}
We mention that an easy check suggests that we could 
have taken remark \ref{fiblen} as a defining condition for 
$\Gamma$ (at least in that case).

As a matter of fact we
have moreover the following analogue of yet another classical property:
orthogonality condition of the base and the fibers:
\begin{lemma}
The Dirac operator  $D_A$ with $A_3\!=\!0$ and 
$\Gamma=\sigma^3$ satisfy 
\begin{equation}
\ncint \Gamma \rho |D_A|^{-3} = 0, \quad  \forall \rho \in \Omega^1_D(\T^2_\theta).
\label{ortho}
\end{equation}
In fact, $\Gamma=\sigma^3$ is a unique (up to sign) $U(1)$-invariant $\Z_2$-grading of 
the Hilbert space, which commutes with the algebra and its commutant and satisfies 
(\ref{ortho}).
\end{lemma}
\begin{proof}
Take an arbitrary one form $\Gamma = \sum \sigma^i \rho_i$. From the 
orthogonality requirement (\ref{ortho}), taking as $\eta = b \sigma^1$ and $b \sigma^2$ 
respectively, for $b \in \T^2_\theta$ we obtain:
$$ \ncint \rho_1 b |D|^{-3} =  \ncint \rho_2 b |D|^{-3} = 0, $$
as in the noncommutative integral we can use $D$ instead of $D_A$. As this
holds for any $b$ then $\rho_1 = \rho_2=0$. Therefore, the only nonvanishing
coefficient is $\rho_3$ and since $\rho_3^2=1$ we recover $\Gamma = \sigma^3$ 
(up to the sign, of course).
\end{proof}

\subsection{Compatible strong connections over $\T^3_\theta$}

Next, we turn to the space of possible connections $\omega$. For simplicity,
we consider the usual unperturbed Dirac operator $D$ given by (\ref{dirtor}) over 
$\T^3_\theta$ (i.e. we take $A_i\!=\!0, i=1,2,3$). 

Here is the full characterization of the possible connections, according to 
the definition (\ref{defcon}):

\begin{lemma}
A $U(1)$ connection over $\T^3_\theta$ with the $U(1)$ action defined 
as in (\ref{u1act}) is a one-form:
\begin{equation}
\omega = \sigma^3 + \sigma^2 \omega_2 + \sigma^1 \omega_1, 
\end{equation}
where $\omega_1,\omega_2 \in \T^2_\theta$ are $U(1)$-invariant 
elements of the algebra $\T^3_\theta$. Every such connection is 
strong. 
\end{lemma}

\begin{proof}
Any $U(1)$-invariant one-form could be written as $\sum_i \sigma^i \omega_i$, 
where all $\omega_i$ are from $\T^2_\theta$. Since 
$\sigma^3 = U_3^{-1} [D,U_3]$, from the condition related to $\delta$ 
we obtain $\omega_3=1$. 
\end{proof}
Finally, we give the Dirac operators compatible according to definition
\ref{compconn} with an arbitrary selfadjoint connection $\omega$ on the 
noncommutative three-torus.
\begin{lemma}
For any selfadjoint connection $\omega$ the compatible Dirac operator has the form:
\begin{equation}
D_{(\omega)} = D - (\sigma^2 J \omega_2 J^{-1}+ \sigma^1 J \omega_1 J^{-1}) \delta_3.
\label{ndirac}
\end{equation}
\end{lemma}

The proof follows straight from direct calculation based on the proof of the
proposition \ref{mpro}. We have an immediate corollary: 

\begin{cor}
The only connection, compatible with the fully ($U(1)^3$) equivariant 
Dirac operator (\ref{dirtor}) on the noncommutative three-torus 
$\T^3_\theta$ is: $\omega=\sigma^3$.
\end{cor}

Finally, observe that the new family of Dirac operators gives a class of new,
spectral geometries over the noncommutative torus. Although they are not real, 
as $D_{(\omega)}$ is not compatible with the real structure, we needed a 
{\em real spectral triple} as a a background geometry providing us with necessary 
tools (in particular, the differential calculus).

The properties of the new class of Dirac operators are yet to be investigated, 
in particular, their spectral properties. However, we expect that they will 
correspond to some locally non-flat geometries. 

\section{Conclusions}

We have attempted to reconstruct the compatibility conditions between the metric
geometry (as given by the Dirac operator) and connections (as defined in the 
algebraic setup for the Hopf-Galois extensions) on noncommutative $U(1)$ bundles.

Although this project has still to be further developed, we have encountered several 
new and interesting phenomena. First of all, we observe that the existence of real
spectral triples is necessary to provide a kind of background geometry. Furthermore,
there are many compatibility conditions, which are necessary to impose on the spectral 
triple, apart from the simple requirement of $U(1)$-equivariance. We have demonstrated
that in the case, where all assumptions are met it is possible to consistently extend
the algebraic definition of strong connections to the case of differential calculi given 
by Dirac operators. 

The requirement of compatibility condition between the connection and the metric (given
implicitly by the Dirac operator) has led us to the discovery of a new family of Dirac
operators. This indicates that these are not objects introduced {\em ad hoc} but have 
a deeper geometrical meaning. We postpone the study of their properties to future work.



\begin{thebibliography}{99}

\bibitem{Amman}
B.Ammann,
{\em The Dirac Operator on Collapsing Circle Bundles}
S\'em. Th. Spec. G\'eom Inst. Fourier Grenoble 16, 33-42 (1998)

\bibitem{AB}
B.Ammann and Ch.B\"ar,
{\em The Dirac Operator on Nilmanifolds and Collapsing Circle Bundles},
Ann. Global Anal. Geom. 16, no. 3, 221-253, (1998)

\bibitem{BM}
T.Brzezinski, S.Majid,
{\em Quantum differentials and the q-monopole revisited}
Acta Appl. Math., 54 (1998) 185--233.

\bibitem{BM2}
T.Brzezinski, S.Majid,
{\it Quantum geometry of algebra factorisations and coalgebra bundles},
Commun. Math. Phys. 213 (2000) 491--521.

\bibitem{Co1}
{A.Connes}, {\it Noncommutative Geometry}, Academic Press, 1994.

\bibitem{Co2}
{A.Connes},
{\it Gravity coupled with matter and foundation of non-commutative
geometry},
Comm. Math. Phys. 182 (1996) 155--176.

\bibitem{CoMo}
A. Connes and H. Moscovici,
{\em The local index formula in noncommutative geometry},
Geom. Funct. Anal.  {\bf 5}  (1995), 174--243.


\bibitem{DGH}
L.D\k{a}browski, H.Grosse, P.M.Hajac,
{\em Strong connections and Chern-Connes pairing in the Hopf-Galois theory},
Communications in Mathematical Physics, 2001, vol.220, 301-331

\bibitem{MCC}
D. Essouabri, B. Iochum, C. Levy and A. Sitarz,
{\em Spectral action on noncommutative torus}, 
J. Noncommut.Geom. {\bf 2} (2008), 53--123.

\bibitem{GLP}
P.Gilkey, J.Leahy, J.H. Park,
{\em The spectral geometry of the Hopf fibration},
Journal Physics A {\bf 29} (1996), 5645--5656

\bibitem{BFV}
J.M.Gracia-Bond\'{\i}a, J.C.V\'arilly, H.Figueroa,
{\it Elements of Noncommutative Geometry},
Birkh\"auser Adv.~Texts, Birkh\"auser, Boston, MA, 2001.

\bibitem{Haj}
P.M.Hajac, 
{\em Strong connections on quantum principal bundles.},
Comm. Math. Phys. {\bf 182} (1996), 579-617

\bibitem{Haj2}
P.M.Hajac, 
{\em A Note on First Order Differential Calculus on Quantum Principal Bundles.},
Czechoslovak J. Phys. {\bf 47}, 1139-1144, (1997).

\bibitem{HM}
P.M.Hajac, S.Majid,
{\em Projective module description of the q-Monopole},
Comm. Math. Phys. 1999, {\bf 206}, 247-264  (1999).

\bibitem{ISS}
B.Iochum, T.Schucker, C.Stephan,
{\em On a Classification of Irreducible Almost Commutative Geometries},
J.Math.Phys. 45 (2004) 5003-5041

\bibitem{Wor}
{S.L.Woronowicz},
{\it Twisted $SU(2)$ group. An example of a noncommutative differential calculus},
Publ.\ RIMS, Kyoto University, 23 (1987) 117--181.

\end{thebibliography}
\end{document}